\documentclass[letterpaper, 10 pt, conference]{ieeeconf}  

\IEEEoverridecommandlockouts                              
\usepackage{amsmath}
\usepackage{amssymb}
\usepackage{graphics,graphicx}
\usepackage{verbatim,color}
\usepackage{subfig}
\usepackage{float}
\usepackage[lined,boxed,commentsnumbered]{algorithm2e}
\usepackage{tikz}

\providecommand{\norm}[1]{\left\lVert#1\right\rVert}
\providecommand{\Z}{\mathbb{Z}} 

\newtheorem{definition}{Definition}
\newtheorem{lemma}{Lemma}
\newtheorem{corollary}{Corollary}
\newtheorem{theorem}{Theorem}
\newtheorem{proposition}{Proposition}

\title{\LARGE \bf
A Game-theoretic Formulation of the Homogeneous Self-Reconfiguration Problem
}

\author{Daniel Pickem$^{1}$, Magnus Egerstedt$^{2}$, and Jeff S. Shamma$^{3}$
\thanks{This research was sponsored by AFOSR/MURI Project \#FA9550-09-1-0538 and ONR Project \#N00014-09-1-0751.}
\thanks{$^{1}$D. Pickem is a Ph.D. student in Robotics,
        Georgia Institute of Technology, Atlanta, USA
        {\tt\small daniel.pickem@gatech.edu}}%
\thanks{$^{2}$M. Egerstedt is with the Faculty of Electrical and Computer Engineering,
        Georgia Institute of Technology, Atlanta, USA
        {\tt\small magnus@gatech.edu}}%
\thanks{$^{3}$J. S. Shamma is with the School of Electrical and Computer Engineering, 
	      Georgia Institute of Technology, Atlanta, USA {\tt\small shamma@gatech.edu}, 
	      and with King Abdullah University of Science and Technology (KAUST), Thuwal, Saudi Arabia {\tt\small jeff.shamma@kaust.edu.sa.}}%
}

\begin{document}
\maketitle
\thispagestyle{empty}
\pagestyle{empty}

\begin{abstract}
In this paper we formulate the homogeneous two- and three-dimensional self-reconfiguration problem over discrete grids as a constrained potential game. We develop a game-theoretic learning algorithm based on the Metropolis-Hastings algorithm that solves the self-reconfiguration problem in a globally optimal fashion. Both a centralized and a fully distributed algorithm are presented and we show that the only stochastically stable state is the potential function maximizer, i.e. the desired target configuration. These algorithms compute transition probabilities in such a way that even though each agent acts in a self-interested way, the overall collective goal of self-reconfiguration is achieved. Simulation results confirm the feasibility of our approach and show convergence to desired target configurations.
\end{abstract}

\IEEEpeerreviewmaketitle

\section{Introduction}
\label{sec:introduction}
Self-reconfigurable systems are comprised of individual agents which are able to connect to and disconnect from one another to form larger functional structures. These individual agents or modules can have distinct capabilities, shapes, or sizes, in which case we call it a heterogeneous system (for example \cite{Fitch2003}). Alternatively, modules can be identical and interchangeable, which describes a homogeneous system (see \cite{Kotay2005}). In this paper, we will present algorithms that reconfigure homogeneous systems and treat self-reconfiguration as a two- and three-dimensional coverage problem.

Self-reconfiguration is furthermore understood to solve the following problem. Given an initial geometric arrangement of cubes (called a configuration) $\mathcal{C}_I$ and a desired target configuration $\mathcal{C}_T$, the solution to the self-reconfiguration problem is a sequence of primitive cube motions that reshapes/reconfigures the initial into the target configuration (see Fig. 1). By configuration we mean a geometric arrangement of a set of agents. The problem setup is then the following. 
\begin{itemize}
 \item The environment $\mathcal{E}$ is a finite two- or three-dimensional discrete grid, i.e. $\mathcal{E} \subseteq \Z^2$ or $\mathcal{E} \subseteq \Z^3$. 
 \item $N$ agents $P = \{1, 2, \dots , N\}$ move in discrete steps through that grid.
 \item Each agent has a restricted action set $\mathcal{R}_i$ which contains only a subset of all its possible actions $A_i$.
 \item An agent's utility or reward $U_i(a \in A)$ is inversely proportional to the distance to the target configuration.
\end{itemize}
Many approaches to self-reconfiguration have been presented in the literature, each with certain short-comings. There have been centralized solutions (both in planning and execution, e.g. \cite{Rus2001}), distributed solutions that required a large amount of communication (see \cite{Fitch2003}) or precomputation (see \cite{Fitch2008}, \cite{Pickem2012}), approaches that were either focused on locomotion (e.g. \cite{Butler2004}) or functional target shape assemblies (e.g. \cite{Kurokawa2008}). Distributed approaches have often relied on precomputation of rulesets (see \cite{Pickem2012}), policies (e.g. \cite{Fitch2008}), or entire sets of paths/folding schemata of agents (e.g. \cite{Cheung2011}).

In this paper, we present a fully decentralized approach to homogeneous self-reconfiguration for which no central decision maker is required. Currently, each agent needs to know the location and shape of the target configuration to compute the utility function of its actions and choose an action. However, no (in the two-dimensional case) or limited communication (in the three-dimensional case) is required to successfully complete a reconfiguration.

The rest of this paper is organized as follows. Section \ref{sec:relatedWork} discusses relevant related work. Section \ref{sec:theoreticalModel} presents the system setup and theoretical formulation of the problem. In Section \ref{sec:deterministicCompleteness} we discuss the completeness of deterministic reconfiguration, which is used in Section \ref{sec:stochasticReconfiguration} to prove the existence of a unique potential function maximizer. In Section \ref{sec:decentralizedAlgorithm} we decentralize the stochastic algorithm and present simulation results in Section \ref{sec:implementation}. Section \ref{sec:conclusion} concludes the paper.

\begin{figure*}
\begin{center}
\includegraphics[width=1.0\textwidth]{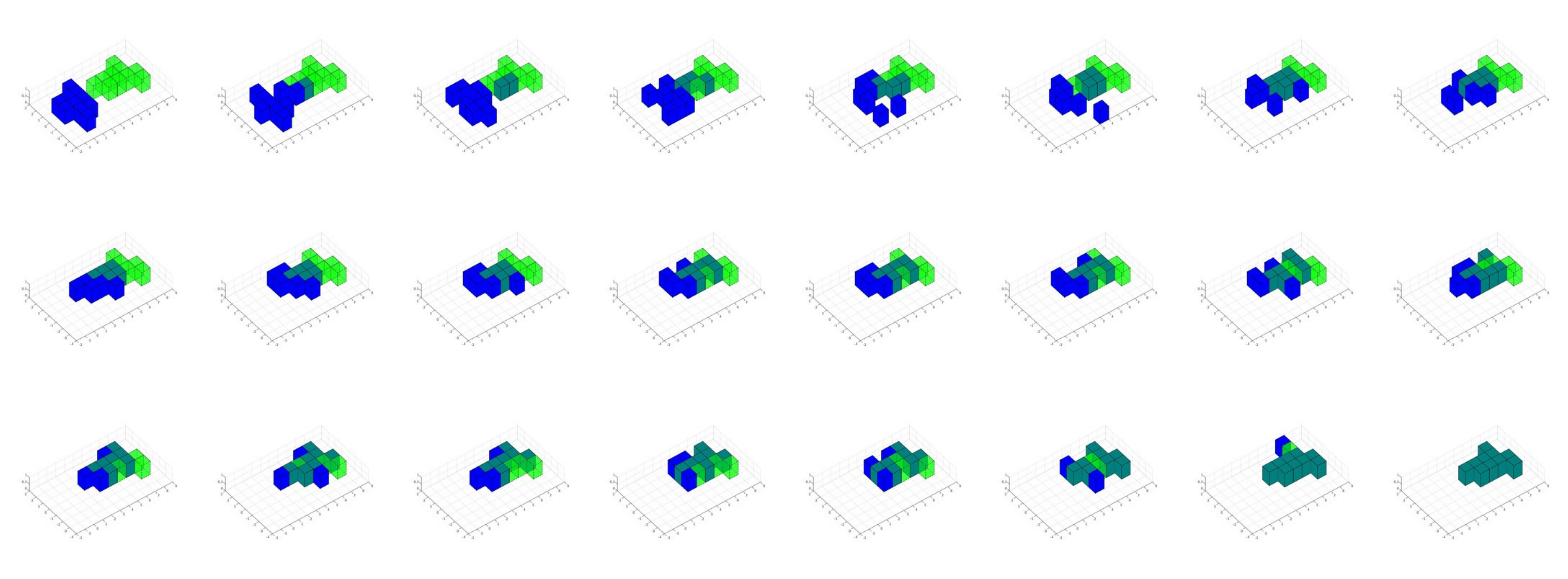}
\end{center}
\caption[Caption]{Example of self-reconfiguration sequence from a random 2D configuration (left) to another random 2D configuration (right). Approximately every $20^{\mathrm{th}}$ time step is shown.}
\label{fig:reconfigurationExample}
\end{figure*}
%
\section{Related Work}
\label{sec:relatedWork}
In this section we want to highlight decentralized approaches in the literature that bear resemblance to methods presented in this paper. Especially relevant to the presented results in this paper are homogeneous self-reconfiguration approaches such as \cite{Butler2004}, and \cite{Kotay2005}, which employ cellular automata and manually designed local rules to model the system. Similar work in \cite{Fitch2008} shows an approach for locomotion through self-reconfiguration represented as a Markov decision process with state-action pairs computed in a distributed fashion. The presented algorithms are decentralized but only applicable to locomotion and not the assembly of arbitrary configurations. 

Precomputed rules are also used in \cite{Pickem2012}, in which graph grammars for self-reconfiguration are automatically generated. Similarly, in \cite{Fox2012} a decentralized approach is presented based on graph grammatical rules, in which automatically generated graph grammars are used to assemble arbitrary acyclic target configurations. Whereas these approaches are able to assemble arbitrary target configurations, they rely on precomputing rulesets for every target configuration.

The algorithms presented in this paper are inspired by the coverage control literature, specifically game-theoretic formulations such as \cite{Arslan2007}, \cite{Lim2011}, and \cite{Zhu2013}. Both static sensor coverage as well as dynamic coverage control with mobile agents are discussed. Note, however, that agents in these papers are limited to movement in two dimensions and operate with a different motion model and most importantly different constraints.

We address these problems with a decentralized algorithm that does not rely on precomputed rulesets, can be used for locomotion as well as the assembly of arbitrary two- and three-dimensional shapes, can handle changing environment conditions as well as changing target configurations, and is scalable to large number of modules due to its decentralized nature.

\section{Problem Formulation}
\label{sec:theoreticalModel}
In this work, we represent agents as cubic modules that move through a discrete lattice or environment $\mathcal{E} = \Z^d$ in discrete steps.\footnote{Since we present two- and three-dimensional self-reconfiguration, throughout this paper, the dimensionality $d$ will be $d \in \{2,3\}$.} Without loss of generality these cubes have unit dimension. Therefore, an agent's current state or action $a_i$ (in a game theoretic sense) is an element $a_i \in \Z^d$. Note that an agent's action is equivalent to its position in the lattice. Cubes can be thought of as geometric embeddings of the agents in our system. A collection of agents is furthermore called a \textit{configuration}. Therefore, a configuration $\mathcal{C}$ composed of $N$ agents is a subset of the representable space $\mathbb{Z}^{dN}$ (see \cite{Pickem2013}). Moreover, we will deal with homogeneous configurations, in which all agents have the same properties and are completely interchangeable.
\subsection{Motion Model}
\label{sec:motionModel}
In the sliding cube model (see \cite{Butler2001}, \cite{Butler2004}, \cite{Pickem2012}, \cite{Pickem2013}, \cite{Rus2001}), a cube is able to perform two primitive motions - a sliding and a corner motion. In general, a motion specifies a translation along coordinate axes and is represented by an element $m \in \Z^d$. A sliding motion is characterized by $\norm{m_s}_{L_1} = 1$, i.e. $m_{s,i} = 1$ for one and only one coordinate $i \in 1, \dots d$, which translates a cube along one coordinate axis. A corner motion on the other hand is defined by $\norm{m_c}_{L_1} = 2$ such that $m_{c,i} = 1$ for exactly two coordinates $i \in 1, \dots d$, which translates a cube along two dimensions.

\subsection{Game-theoretic Formulation}
In this section we formulate the homogeneous self-reconfiguration as a potential game (see \cite{Monderer1996}), which is a game structure amenable to globally optimal solutions. Generally, a game is specified by a set of players $i \in P = \{1, 2, \dots, N\}$, a set of actions $\mathcal{A}_i$ for each player, and a utility function $U_i(a) = U_i(a_i, a_{-i})$ for every player $i$. In this notation, $a$ denotes a joint action profile $a = (a_1, a_2, \dots, a_N)$ of all $N$ players, and $a_{-i}$ is used to denote the actions of all players other than agent $i$. \\
In a constrained game, the actions of agents are constrained through their own and other agents' actions. In other words, given an action set $\mathcal{A}_i$ for agent $i$, only a subset $\mathcal{R}_i(a)) \subset \mathcal{A}_i$ is available to agent $i$. A constrained potential game is furthermore defined as follows.
%
\begin{definition}
\label{def:constrained_game}
 A \textit{constrained exact potential game} (see \cite{Zhu2013})  is a tuple $G = (\mathcal{P}, \mathcal{A},$ $\{U_i(.)\}_{i \in \mathcal{P}}, \{R_i(.)\}_{i \in \mathcal{P}}, \Phi(A))$, where 
 \begin{itemize}
  \item $\mathcal{P} = \{1, \dots, N\}$ is the set of $N$ players
  \item $\mathcal{A} =  \mathcal{A}_1 \times \dots \times \mathcal{A}_N$ is the product set of all agents' action sets $\mathcal{A}_i$
  \item $U_i: A \rightarrow \mathbb{R}$ are the agents' individual utility functions
  \item $R_i: \mathcal{A} \rightarrow 2^{\mathcal{A}_i}$ is a function that maps a joint action to a restricted action set for agent $i$
 \end{itemize}
 Additionally, the agents' utility functions are aligned with a global objective function or potential $\Phi: \mathcal{A} \rightarrow \mathbb{R}$ if for all agents $i \in \mathcal{P}$, all actions $a_i, a_i' \in R_i(a)$, and actions of other agents $a_{-i} \in \prod_{j \ne i} \mathcal{A}_j$ the following is true
 \[
  U_i(a_i', a_{-i}) - U_i(a_i, a_{-i}) = \Phi(a_i', a_{-i}) - \Phi(a_i, a_{-i})
 \]
\end{definition}

The last condition of Def. \ref{def:constrained_game} implies an alignment of agents' individual incentives and the global goal. Therefore, under unilateral change (only agent $i$ changes its action from $a_i$ to $a_i'$) the change in utility for agent $i$ is equivalent to the change in the global potential $\Phi$. This is a highly desirable property since the maximization of all agents' individual utilities yields a maximum global potential. 
We can now formulate the self-reconfiguration problem in game theoretic terms and show that it is indeed a constrained potential game. 
\begin{definition}
\label{def:self_reconfiguration}
\textit{Game theoretic self-reconfiguration} can be formulated as a constrained potential game, where the individual components are defined as follows:
\begin{itemize}
 \item The set of players $\mathcal{P} = \{1, 2, \dots , N\}$ is the set of all $N$ agents in the configuration.
 \item The action set of each agent $\mathcal{A}_i = \Z^d$ is a set of discrete lattice positions (or a finite or infinite subset of $\Z^d$).
 \item The utility function of each agent is $U_i(a) = \frac{1}{\mathrm{dist}(a_i, \mathcal{C}_T) + 1}$. Here, $\mathcal{C}_T$ is the target configuration and $\mathrm{dist}(a_i, \mathcal{C}_T) = \min_{a_j \in \mathcal{C}_T} \norm{a_i - a_j}$.
 \item The restricted action sets $R_i(a \in A)$ are computed according to Section \ref{sec:actionSetComputation}.
 \item The global potential $\Phi(a \in A) = \displaystyle \sum_{i \in \mathcal{P}} U_i(a)$.
\end{itemize}
Note that the utility of an agent is independent of all other agents' actions and depends exclusively on its distance to the target configuration. An agent's action set, however, is constrained by its own as well as other agents' actions. The goal of the game theoretic self-reconfiguration problem is to maximize the potential function, i.e. 
\[
 \displaystyle \max_{a \in \mathcal{A}} \Phi(a) = \max_{a \in \mathcal{A}} \sum_{i \in \mathcal{P}} U_i(a)
\]
\end{definition}
This can be interpreted as a coverage problem where the goal of all agents is to cover all positions in the target configuration. Therefore maximizing the potential is equivalent to maximizing the number of agents that cover target positions $a_i \in \mathcal{C}_T$. The following propositions shows that this formulation indeed yields a potential game. 
%
\begin{proposition}
\label{prop:potentialGame}
The self-reconfiguration problem in Def. \ref{def:self_reconfiguration} constitutes a constrained potential game with $\Phi(a) = \displaystyle \sum_{i \in \mathcal{P}} U_i(a)$ and $U_i(a) = \frac{1}{\mathrm{dist}(a_i, \mathcal{C}_T) + 1}$.
\end{proposition}

\begin{proof}
Let the agents' utility functions $U_i(a_i, a_{-i}) = \Phi(a_i, a_{-i}) - \Phi(a_i^0, a_{-i})$, where $a_i^0$ denotes the null action of agent $i$, which is equivalent to removing agent $i$ from the environment. Then an agent's utility is its marginal contribution to coverage (Wonderful Life Utility, see \cite{Arslan2007}), or in other words, the grid cells covered exclusively by agent $i$. But since each agent covers exactly the grid cell it currently occupies, the following holds. 
\begin{eqnarray}
  U_i(a_i, a_{-i}) &=& \Phi(a_i, a_{-i}) - \Phi(a_{i}^{0}, a_{-i}) \nonumber \\
		   &=& \sum_{j \in \mathcal{P}} U_j(a) - \sum_{j \in \mathcal{P} \setminus \{i\}} U_j(a) = U_i(a) \nonumber
\end{eqnarray}
Therefore, 
\begin{eqnarray}
  U_i(a_i', a_{-i}) - U_i(a_i, a_{-i}) &=& \Phi(a_i', a_{-i}) - \Phi(a_i, a_{-i}) \nonumber
\end{eqnarray}
\end{proof}

As we will see in Section \ref{sec:decentralizedAlgorithm}, this potential game structure allows us to derive a decentralized version of the presented learning algorithm.

\subsection{Action Set Computation}
\label{sec:actionSetComputation}
A core component of constrained potential games is the computation of restricted action sets. Unlike previous work (see for example \cite{Marden2012} and \cite{Zhu2013}), agents in our setup are constrained not just by their own actions, but also those of others. In this section we present methods for computing restricted action sets such that agents comply with motion constraints as well as constraints imposed by other agents.
%
\paragraph{2D reconfiguration}
In the two-dimensional case agents are restricted to motions on the xy-plane. Unlike in previous work (see \cite{Pickem2012} and \cite{Pickem2013}) where we required a configuration to remain connected at all times, in this work, agents are allowed to disconnect from all (or a subset of) other agents. This approach enables agents to separate from and merge with other agents at a later time. To formalize this idea, we first review some graph theoretic concepts.
\begin{definition}
\label{def:connectivity_graph}
 Let $G = (V, E)$ be the graph composed of $N$ nodes with $V = \{v_1, v_2, \dots, v_N\}$, where node $v_i$ represents agent or location $i$. Then $G$ is called the \textit{connectivity graph of configuration $\mathcal{C}$} if $E = V \times V$ with $e_{ij} \in E$ if $\norm{a_i - a_j} = 1$.
\end{definition}
This definition implies that two nodes $v_i, v_j$ in the connectivity graph are adjacent, if agent or location $i$ and $j$ are located in neighboring grid cells. Note that a connectivity graph can be computed for any set of grid positions, whether these positions are occupied by agents or not. We furthermore use the notions of paths on graphs and graph connectivity in the usual graph theoretic sense. 
Note that $G$ is not necessarily connected as (groups of) agents can split off. Therefore, $G$ generally consists of connected components $C_i$ such that $G = \{C_1, C_2, \dots, C_m\}$. Since the edge set $E$ of $G$ is time-varying, the number $m$ of connected components changes with time as well. Based on the connectivity graph $G$ and the current joint action, we now define the function $R_i: \mathcal{A} \rightarrow 2^{\mathcal{A}_i}$, which maps from the full joint action set to a restricted action set for agent $i$ and is based on the following two definitions of sets of primitive actions.

\begin{definition}
\label{def:motion_set_sliding}
 The set of all currently possible \textit{sliding motions} is $\mathcal{M}_s = \left\{a_i' \in \Z^d \setminus a_{-i}: \norm{m_s}_{L_1} = 1\right\}$, where $m_s = a_i' - a_i$.
\end{definition}

\begin{definition}
\label{def:motion_set_corner}
 The set of all currently possible \textit{corner motions} is $\mathcal{M}_c = \left\{a_i' \in \Z^d \setminus a_{-i}: \norm{m_c}_{L_1} = 2 \; \land \; m_{c,j} \in \{0,1\}\right\}$, where $j \in [1, \dots , d]$ and $m_c = a_i' - a_i$.
\end{definition}
Note that $\mathcal{M}_s$ and $\mathcal{M}_c$ in Def. \ref{def:motion_set_sliding} and Def. \ref{def:motion_set_corner} are equally applicable to 2D and 3D. 
These definitions encode the motion model outlined in Section \ref{sec:motionModel} and allow us to define the restricted action set in two dimensions as follows.
\begin{definition}
\label{def:actionSet2D}
  The two-dimensional restricted action set is given by $R_{i}^{2D}(a) = \mathcal{M}_{s} \cup \mathcal{M}_{c}$.
\end{definition}
This definition ensures that agent $i$ can only move to unoccupied neighboring grid positions $a_i'$ through sliding or corner motions (or stay at its current position $a_i$ - see Algorithm \ref{alg:global} and Algorithm \ref{alg:local}). All other agents replay their current actions $a_{-i}$.
%
%
\paragraph{3D reconfiguration}
Whereas in the 2D case agents were allowed to move to all unoccupied neighboring grid cell regardless of connectivity constraints, in the three-dimensional case we introduce the requirement of \textit{groundedness}. An agent is immobile, if executing an action would remove groundedness from any of its neighbors. Groundedness requires a notion of ground plane, which is defined as follows.

\begin{definition}[Ground Plane]
\label{def:ground_plane}
 The \textit{ground plane} is the set $S_{GP} = \{s \in \mathcal{E}: s_z = 0 \}$ where $\mathcal{E} \subseteq \Z^3$ and the corresponding connectivity graph is $G_{GP} = (V_{GP}, E_{GP})$ with $e_{ij} \in E_{GP}$ if $\norm{s_i - s_j} = 1$.
\end{definition}
Note that the ground plane is defined as the xy-plane and its connectivity graph $G_{GP}$ is, by definition, connected. Positions $s \in S_{GP}$ are not allowed to be occupied by agents, therefore $a_i \in A_i \setminus S_{GP}$ $\forall i \in \mathcal{P}$. Using the graph $G_{GP}$, we define $G' = (V', E')$ as $V' = V \cup V_{GP}$ and $E' = V' \times V'$ such that $e_{ij} \in E'$ if for $v_i, v_j \in V'$ we have $\norm{a_i - a_j}_{L_1} = 1$. Note that $G'$ represents the current configuration including the ground plane, and $a_i$ represents an action of an agent or an unoccupied position in the ground plane.
\begin{definition}[Groundedness]
\label{def:groundedness}
 An agent $i$ is \textit{grounded} if there exists a path on $G'$ from $v_i \in V \subset V'$ to some $v_k \in V_{GP} \subset V'$, where $v_i$ represents agent $i$ in the connectivity graph $G$ (see Def. \ref{def:connectivity_graph}).
 A configuration $\mathcal{C}$ is \textit{grounded} if every agent $i \in \mathcal{P}$ is grounded.
\end{definition}
The idea behind groundedness hints at an embedding of a self-reconfigurable system in the physical world, where agents cannot choose arbitrary positions in the environment (e.g. float in free space). More importantly, we use the notion of groundednes to prove completeness of deterministic reconfiguration in Section \ref{sec:deterministicCompleteness} and irreducibility of the underlying Markov chain in Section \ref{sec:stochasticReconfiguration}.\\
An agent can verify groundedness in a computationally cheap way through a depth-first search, which is complete and guaranteed to terminate in time proportional to $O(N)$ in a finite space. The notion of groundedness also informs the restricted action set computation. If all neighbors $\mathcal{N}_{i} = \{v_j \in V: e_{ij} \in E\}$ (adjacency according to $G$ in Def. \ref{def:connectivity_graph}) of agent $i$ can compute an alternate path to ground (other than through agent $i$) then agent $i$ is allowed to move. 
To formalize this idea, let $G_{-i} = (V_{-i}, E_{-i})$ with $V_{-i} = V \cup V_{GP} \setminus \{v_i\}$ and $E_{-i} = V_{-i} \times V_{-i}$ such that $e_{ij} \in E_{-i}$ if for $v_i, v_j \in V_{-i}$ we have $\norm{a_i - a_j}_{L_1} = 1$. $G_{-i}$ is therefore the connectivity graph of the current configuration including the ground plane without agent $i$. Subsequently, $R_{i}^{3D}(a)$ is defined as follows. 
\begin{definition}
\label{def:actionSet3D}
 The three-dimensional restricted action set $R_{i}^{3D}(a) = \mathcal{M}_{s} \cup \mathcal{M}_{c}$ if all agents $v_j \in \mathcal{N}_i$ are grounded on $G_{-i}$. Otherwise, $R_{i}^{3D}(a) = \{a_i\}$.
\end{definition}
This definition encodes the same criteria as the two-dimensional action set with the additional constraint of maintaining groundedness (see Fig. \ref{fig:groundedness}). If agent $i$ executing an action would leave any of its neighbors ungrounded, agent $i$ is not allowed to move. 

\begin{figure}
\centering
\subfloat[A movement of agent $c_1$ would remove groundedness of the agent $c_2$.]{
    \begin{tikzpicture}[baseline,scale=0.1]
      \node[anchor=south west, inner sep=0] (image) at (0,0){\includegraphics[width=0.18\textwidth]{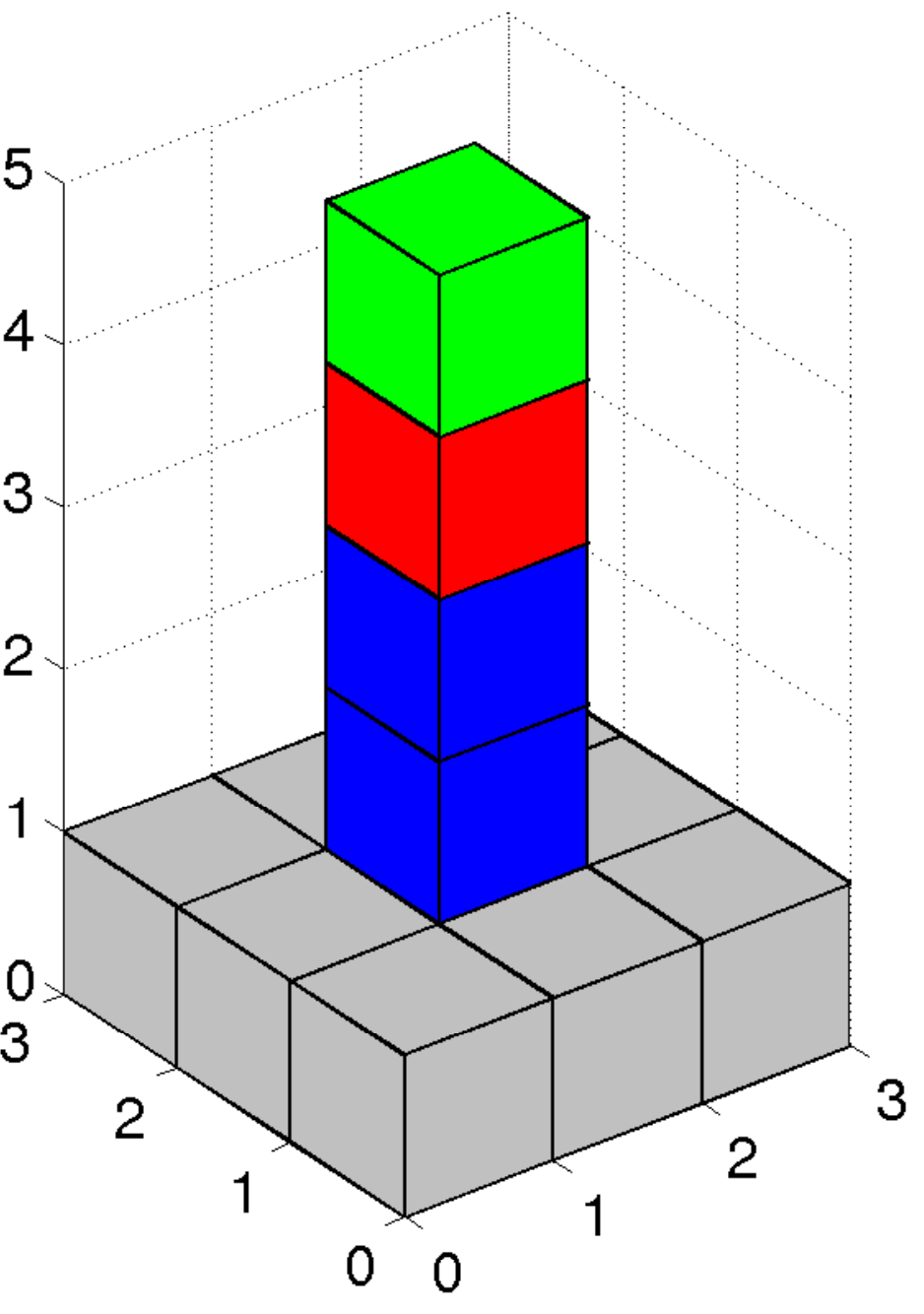}\label{fig:groundedness_1}};
      \begin{scope}[x={(image.south east)},y={(image.north west)}]
	\node[anchor=west, thick] at (0.44,0.97){$\mathcal{C}$};
	\draw[red,ultra thick,rounded corners] (0.31,0.28) rectangle (0.68,0.92);
	\node[anchor=west] (ground) at (-0.2,0.15){$S_{GP}$};
	\node[anchor=west] (a) at (0.26,0.23){};
	\node[anchor=west] (c1) at (0.48,0.63){};
	\node[anchor=west] (c2) at (0.48,0.73){};
	\node[anchor=west] (c1Label) at (0.8,0.7){$c_2$};
	\node[anchor=west] (c2Label) at (0.8,0.8){$c_1$};
	\draw[->, thick] (ground) to (a);
	\draw[->, thick] (c1Label) to (c1);
	\draw[->, thick] (c2Label) to (c2);
      \end{scope}
    \end{tikzpicture}
}\hfill{}
\subfloat[Agenet $c_1$ can move without breaking the groundedness constraint for the agents $c_2$ and $c_3$.]{
    \begin{tikzpicture}[baseline,scale=0.1]
      \node[anchor=south west, inner sep=0] at (0,0){\includegraphics[width=0.22\textwidth]{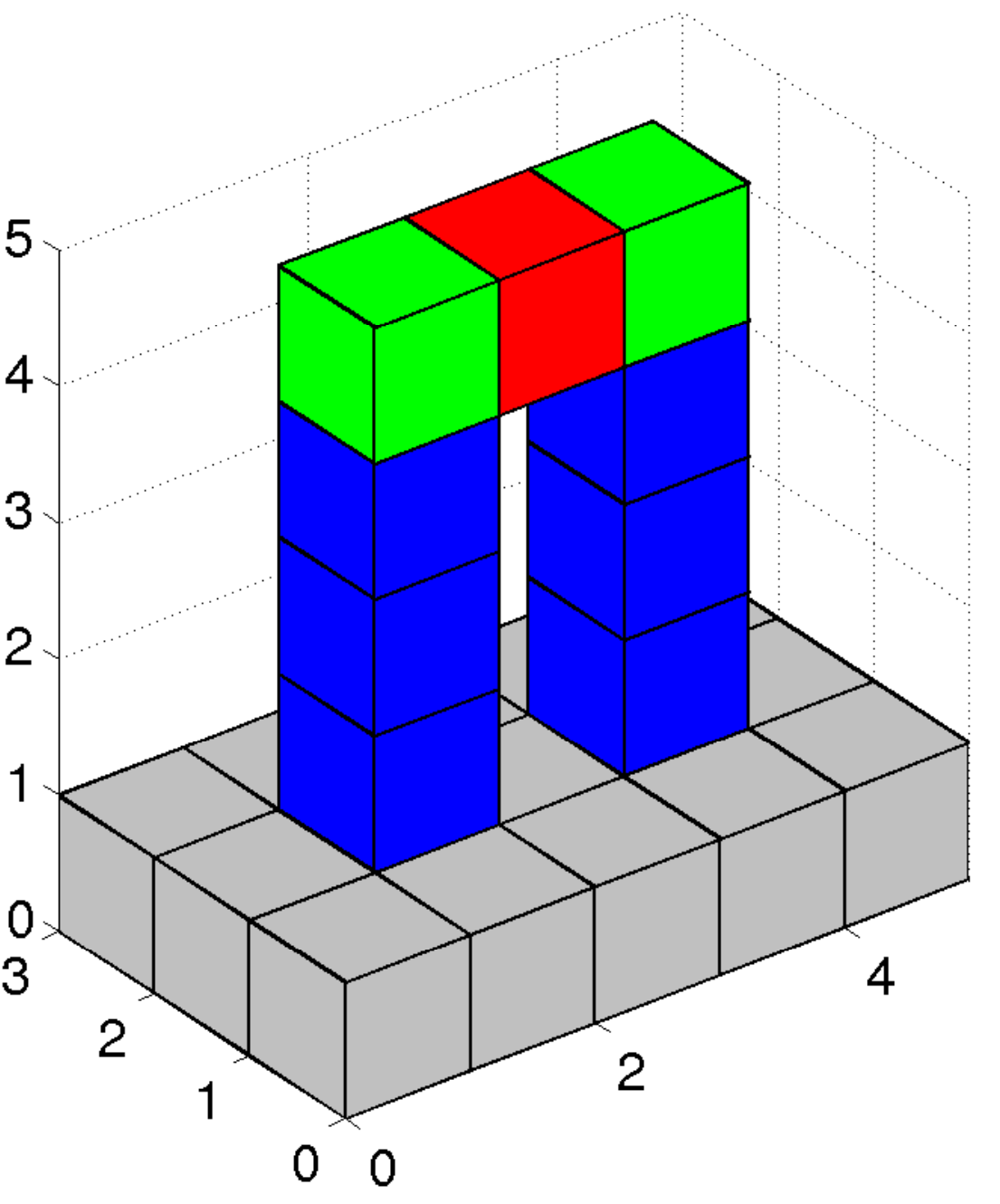}\label{fig:groundedness_2}};
      \begin{scope}[x={(image.south east)},y={(image.north west)}]
	\node[anchor=west, thick] at (0.55,1){$\mathcal{C}$};
	\draw[red,ultra thick,rounded corners] (0.31,0.28) rectangle (0.95,0.95);
	\node[anchor=west] (ground) at (0.0,0.0){$S_{GP}$};
	\node[anchor=west] (a) at (0.26,0.23){};
	\node[anchor=west] (c2) at (0.76,0.79){};
	\node[anchor=west] (c1) at (0.62 ,0.76){};
	\node[anchor=west] (c3) at (0.46,0.71){};
	\node[anchor=west] (c2Label) at (1,0.8){$c_2$};
	\node[anchor=west] (c1Label) at (1,0.7){$c_1$};
	\node[anchor=west] (c3Label) at (1,0.6){$c_3$};
	\draw[->, thick] (ground) to (a);
	\draw[->, thick] (c1Label) to (c1);
	\draw[->, thick] (c2Label) to (c2);
	\draw[->, thick] (c3Label) to (c3);
      \end{scope}
    \end{tikzpicture}
}
\caption[Caption]{Examples of grounded configurations and feasible motions of cubes.}
\label{fig:groundedness}
\end{figure}
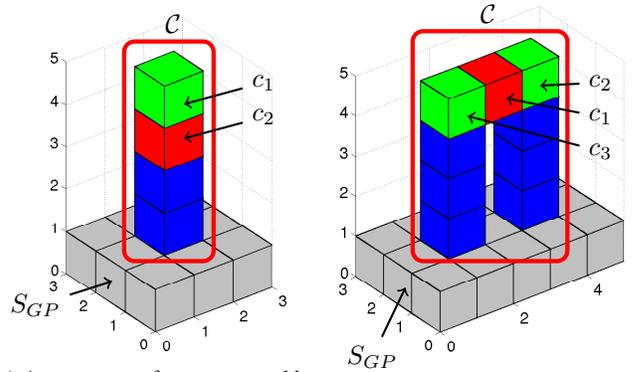
%
\section{Deterministic Completeness}
\label{sec:deterministicCompleteness}
In this section we establish completeness of deterministic reconfiguration in two and three dimensions. We will show that for any two configurations $\mathcal{C}_I$ and $\mathcal{C}_T$ there exists a deterministically determined sequence of individual agent actions such that configuration $\mathcal{C}_I$ will be reconfigured into $\mathcal{C}_T$. These results are required to show irreducibility of the Markov chain induced by the learning algorithm outlined in Section \ref{sec:stochasticReconfiguration}. Irreducibility guarantees the existence of a unique stationary distribution and furthermore a unique potential function maximizer. We first show completeness of 2D reconfiguration.

\begin{theorem}[Completeness of 2D reconfiguration]
\label{theorem:completeness_2D}
 Any given two-dimensional configuration $\mathcal{C}_{I}$ can be reconfigured into any other two-dimensional configuration $\mathcal{C}_{T}$, i.e. there exists a finite sequence of configurations $\{\mathcal{C}_I = \mathcal{C}_0, \mathcal{C}_1, \dots, \mathcal{C}_M = \mathcal{C}_T \}$ such that two consecutive configurations differ only in one individual agent motion.
\end{theorem}

\begin{proof}
 Without loss of generality, assume that $\mathcal{C}_I$ and $\mathcal{C}_T$ do not overlap, i.e. for no $c_i \in \mathcal{C}_I$ is it true that also $c_i \in \mathcal{C}_T$. Additionally, assume that $\mathcal{C}_I$ and $\mathcal{C}_I$ are separated along one dimension $k \in \{e_x, e_y, e_z\}$ (with $e_x, e_y, e_z$ being the basis vectors of the lattice), i.e. $\forall c_{i, k} \in \mathcal{C}_I$ we have that $c_{i, k} < c_{j, k}, \; \forall c_j \in \mathcal{C}_T$.
 Then at each time step $t$, select the agent $i$ whose current position $c_i \in \mathcal{C}$ is closest to an unoccupied position $c_j \in \mathcal{C}_T$. Plan a deterministic path of primitive agent motions $p_i = \{c_i = c_i^0 \rightarrow c_i^1 \rightarrow \dots \rightarrow c_i^m = c_{t_i}\}$ using a complete path planner such as A*. Note that such a path always exists since we don't require agents to remain connected to any other agents. Therefore, the path planning problem is reduced to single agent path planning on a discrete finite grid, which is complete because A* is complete. This greedy selection process of the agent-target pairs together with a complete path planning approach suffices to reconfigure any two-dimensional configuration into any other two-dimensional configuration (similar to flood-fill algorithms).
\end{proof}
The result in Theorem \ref{theorem:completeness_2D} holds for any configuration, even configurations that consist of multiple connected components. Before we can show a similar result for the 3D case we need to introduce a graph theoretic result. 

\begin{lemma}
\label{lemma:articulation_point}
 According to Lemma 6 in \cite{Rus2001}, any finite graph with at least two vertices contains at least two vertices which are not articulation points.\footnote{An articulation point is a vertex in a graph whose removal would disconnect the graph.}
\end{lemma}

\begin{theorem}[Completeness of 3D to 2D reconfiguration]
\label{theorem:intermediate_config}
 Any finite grounded 3D configuration $\mathcal{C}^{G, 3D}$ can be reconfigured into a 2D configuration $\mathcal{C}_{Int}^{2D}$, i.e. there exists a finite sequence of configurations $\{ \mathcal{C}^{G,3D} = \mathcal{C}_0, \mathcal{C}_1, \dots, \mathcal{C}_M = \mathcal{C}_{Int}^{2D} \}$ such that two consecutive configurations differ only in one individual agent motion.
\end{theorem}

\begin{proof}
 Without loss of generality, assume that the connectivity graph of $\mathcal{C}^{G, 3D}$ consists of one connected component. In any finite grounded 3D configuration, there always exists an agent $i \in \mathcal{P}$ with a non-empty restricted action set $|R_i(a)| > 0$.. Agent $i$ is therefore mobile and there exists a finite path of individual agent motions $p_i = \{a_i = a_i^0, a_i^1, \dots, a_i^m\}$ such that for some $s \in S_{GP}$ $\norm{a_i^m, s}_{L_1} = 1$, i.e. the agent's final action is a position on the ground plane $S_{GP}$.
 
 Let the subset of agents $\mathcal{P}_{z>1} \subset \mathcal{P}$ contain those agents whose positions are not on the ground plane and $G_{z>1} = (V_{z>1}, E_{z>1})$ the corresponding connectivity graph. Furthermore, let $G'' = (V'', E'')$ be such that $V'' = V_{z>1} \cup \{v^{GP}\}$, where $v^{GP}$ is a single node representing all the agents on the ground plane and $E''$ such that $e_{ij} \in E''$ if for $v_i, v_j \in V''$ we have $\norm{a_i - a_j}_{L_1} = 1$.
 According to Def. \ref{def:actionSet3D}, an agent $i$ is mobile if it is not an articulation point in the connectivity graph $G''$ (see Fig. \ref{fig:proof_3D_to_2D}). 
 In $G''$, $v^{GP}$ may or may not be a non-articulation points, but according to Lemma \ref{lemma:articulation_point}, in every connected graph there always exist at least two non-articulation points. Therefore, there exists at least one agent $i$ that has a non-empty restricted action set. For that agent, one can compute a deterministic action sequence that moves agent $i$ to the ground plane.
 In other words, at each iteration $t$ we transfer one vertex from $V_{z>1}$ to $v^{GP}$ such that $|V_{z>1}^t| = |V_{z>1}^{t-1}| - 1$ until $|V_{z>1}| = 0$ and all agents have been moved to $v^{GP}$. This process terminates in a finite number of time steps because the initial configuration $\mathcal{C}^{G, 3D}$ is finite. The result is a 2D configuration $\mathcal{C}_{Int}^{2D}$ representing $G''$.
\end{proof}

\begin{figure}
  \begin{center}
    \input{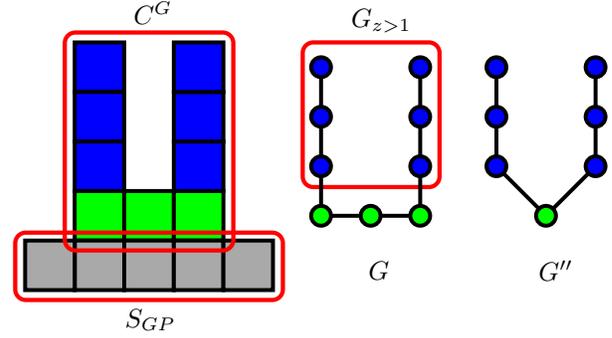} 
    \caption{Example of a grounded configuration $\mathcal{C}^{G}$, the ground plane $S_{GP}$, associated connectivity graph $G$. $G_{z>1}$ represents all agents not on the ground plane, while all agents on the ground plane are represented by a single node in $G''$}
    \label{fig:proof_3D_to_2D}
  \end{center}
\end{figure}

\begin{corollary}
\label{cor:deterministicCompletenes3D}
 Any finite grounded 3D configuration $\mathcal{C}^{G, 3D}_I$ can be reconfigured into any other finite grounded 3D configuration $\mathcal{C}^{G, 3D}_{T}$.
\end{corollary}

\begin{proof}
 Since, according to Theorem \ref{theorem:intermediate_config}, any finite grounded 3D configuration $\mathcal{C}^{G, 3D}_I$ can be reduced to an intermediate 2D configuration $\mathcal{C}_{Int}^{2D}$ in a finite number of steps, the reverse is also true - any finite grounded 3D configuration $\mathcal{C}^{G, 3D}_{T}$ can be assembled from some 2D configuration $\mathcal{C}_{Int}^{2D'}$ in a finite number of steps.
 According to Theorem \ref{theorem:completeness_2D}, any 2D configuration $\mathcal{C}_{Int}^{2D}$ can be reconfigured into any other 2D configuration $\mathcal{C}_{Int}^{2D'}$. Therefore, there exists a deterministic finite action sequence from $\mathcal{C}_{I}^{G, 3D}$ to $\mathcal{C}^{G, 3D}_{T}$.
\end{proof}

\section{Stochastic Reconfiguration}
\label{sec:stochasticReconfiguration}
In this and the following section we present a stochastic reconfiguration algorithm that is fully distributed, does not require any precomputation of paths or actions, and can adapt to changing environment conditions. Unlike log-linear learning (\cite{Blume1993}), which cannot handle restricted action sets, and variants such as binary log-linear learning (\cite{Arslan2007}, \cite{Lim2011}, \cite{Marden2012}), which can only handle action sets constrained by an agent's own previous action, the presented algorithm guarantees convergence to the potential function maximizer even if action sets are constrained by all agents' actions.

Our algorithm is based on the Metropolis-Hastings algorithm (\cite{Metropolis1953},\cite{Hastings1970}), which allows the design of transition probabilities such that the stationary distribution of the underlying Markov chain is a desired target distribution, which we choose to be the Gibbs distribution. This choice enables a distributed implementation of the learning rule in Theorem \ref{theorem:transition_probabilities} through the potential game formalism (see Corollary \ref{cor:decentralizedTransitionProbabilities}). The Metropolis-Hastings algorithm guarantees two results: existence and uniqueness of a stationary distribution. We will use these properties to show that the only stochastically stable state is x*, the potential function maximizer.
 
\begin{theorem}
\label{theorem:transition_probabilities}
 Given any two states $x_i$ and $x_j$ representing global configurations, the transition probabilities 
 $$p_{ij} =  \
 \begin{cases} 
  q_{ji}e^{\frac{1}{\tau}(\Phi(x_j) - \Phi(x_i))} & \mathrm{if} \; e^{\frac{1}{\tau}(\Phi(x_j) - \Phi(x_i))}\frac{q_{ji}}{q_{ij}} \leq 1 \\
  q_{ij} & \mathrm{o.w.}
 \end{cases}$$
 guarantee that the unique stationary distribution of the underlying Markov chain is a Gibbs distribution of the form $Pr[X=x] = \frac{e^{\frac{1}{\tau}\Phi(x)}}{\sum_{x' \in \mathcal{X}} e^{\frac{1}{\tau}\Phi(x')}}$. 
\end{theorem}

\begin{proof}
 Let $\mathcal{X}$ be a finite state space containing all possible states of configurations composed of $N$ agents.\footnote{Such a space is finite if we assume translation and rotation invariance of the configuration as well as a finite environment.} On that state space, let the desired target distribution be $\pi(x) = Pr[X=x] = \frac{e^{\frac{1}{\tau}\Phi(x)}}{\sum_{x' \in \mathcal{X}} e^{\frac{1}{\tau}\Phi(x')}}$ with $\Phi$ defined in Def. \ref{def:self_reconfiguration}. By applying the Metropolis-Hastings algorithm, we can compute transition probabilities $P = \{p_{ij}\}$ such that $\pi$ is the stationary distribution of $P$, i.e. $\pi = \pi P$.
 
 In the Metropolis-Hastings algorithm, a transition probability is represented as $p_{ij} = g(x_i \rightarrow x_j)\alpha(x_i \rightarrow x_j)$, where $g(x_i \rightarrow x_j)$ is the proposal distribution and $\alpha(x_i \rightarrow x_j)$ is the acceptance distribution. Both are conditional probabilities of proposing/accepting a state $x_j$ given that the current state is $x_i$.
 
 Let agent $k$ achieve the transition from state $x_i$ to $x_j$ through action $a_k \in \mathcal{R}_k(a)$. Then one possible choice for the proposal distribution is $g(x_i \rightarrow x_j) = q_{ij} = \frac{1}{|R_k|}, \; \forall j \in \{1,\dots,|R_k|\}$, i.e. a random choice among all available actions of agent $k$. 
 According to Hastings (\cite{Hastings1970}), a popular choice for the acceptance distribution is the Metropolis choice $\alpha_{ij} = \min\left\{1, \frac{\pi_j q_{ji}}{\pi_i q_{ij}} \right\}$. Note that unlike in the original formulation (\cite{Metropolis1953}), we don't assume that $q_{ij} = q_{ji}$ (see Fig. \ref{fig:forward_and_reverse_action_sets} for an illustration of $q_{ij}$ and $q_{ji}$). These choices result in the following transition probabilities:
 
 $$p_{ij} =  \
 \begin{cases} 
    q_{ji} \frac{\pi_j}{\pi_i} & \mathrm{if} \; \frac{\pi_j q_{ji}}{\pi_i q_{ij}} \leq 1 \\
    q_{ij} & \mathrm{o.w.}
 \end{cases}$$
 
 It is easily verified that these $p_{ij}$ satisfy the detailed balance equation $\pi_i p_{ij} = \pi_j p_{ji}$ and thus guarantee the existence of a stationary distribution (see \cite{Metropolis1953} and \cite{Hastings1970}). The resulting $p_{ij}$ follow from the definition of $\pi_i = \frac{e^{\frac{1}{\tau}\Phi(x_i)}}{\sum_{x' \in \mathcal{X}} e^{\frac{1}{\tau}\Phi(x')}}$ and similarly $\pi_j$.
 
 Uniqueness of the stationary distribution follows from the irreducibility of the Markov chain induced by $P = \{p_{ij}\}$, which is the case for our choice of proposal and acceptance distribution because they assign a nonzero probability to every action in any restricted action set. By Theorem \ref{theorem:intermediate_config} and Corollary \ref{cor:deterministicCompletenes3D} we know that any state $x_j$ can be reached from any other state $x_i$ and vice versa. Thus any action path has nonzero probability and irreducibility follows.
\end{proof}

Theorem \ref{theorem:transition_probabilities} applies equally for 2D and 3D configuration. However, for 3D reconfiguration, the proof relies implicitly on the notion of groundedness to show irreducibility of the underlying Markov chain (through the computation of $R_i^{3D}$).
\begin{figure}
  \begin{center}
    \input{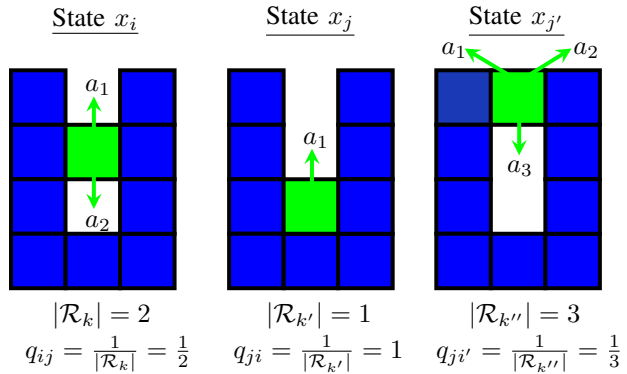}
    \caption{Example of forward and reverse actions with their associated proposal probabilities $q_{ij}$, $q_{ji}$, and $q_{ji}'$. Note that $x_i$, $x_j$, $x_j'$ are states of the entire configuration, and agent $k$ is the currently active agent.}
    \label{fig:forward_and_reverse_action_sets}
  \end{center}
\end{figure}
The following theorem requires the definition of stochastic stability.
\begin{definition}[Stochastic Stability \cite{Young1993}]
 A state $x_i \in \mathcal{X}$ is \textit{stochastically stable} relative to a Markov process $P^{\epsilon}$ if the following holds for the stationary distribution $\pi$ $\lim_{\epsilon \rightarrow 0} \pi^{\epsilon}_{x_i} > 0$.
\end{definition}
Note that the Markov process is defined through the transition probabilities in Theorem \ref{theorem:transition_probabilities} and the stationary distribution is a Gibbs distribution. Furthermore $\epsilon$ is equivalent to the learning rate $\tau$ in the following proof.
\begin{theorem}
\label{theorem:stochastic_stability}
 Consider the self-reconfiguration problem in Def. \ref{def:self_reconfiguration}. If all players adhere to the learning rule in Theorem \ref{theorem:transition_probabilities} then the unique stochastically stable state $x^*$ is the state that maximizes the global potential function. 
\end{theorem}

\begin{proof}
  Note that this result holds by definition of the desired target distribution, which is a Gibbs distribution centered at the state of maximum global potential and defines the probability of being in a state $x$ as $Pr[X=x] = \frac{e^{\frac{1}{\tau}\Phi(x)}}{\sum_{x' \in \mathcal{X}} e^{\frac{1}{\tau}\Phi(x')}}$. The learning rate or temperature $\tau$ represents the willingness of an agent to make a suboptimal choice (i.e. explore the state space). As $\tau \rightarrow 0$, $Pr[X=x] \rightarrow 0$ for all states $x \neq x^*$ which are not potential function maximizers (see \cite{Blume1993} and \cite{Marden2012}).
\end{proof}

Note that the maximum global potential is achieved when all agents are at a target position $a_i \in \mathcal{C}_T$. Algorithm \ref{alg:global} shows an implementation of Theorem \ref{theorem:transition_probabilities}. In Algorithm \ref{alg:global}, similarly to the Metropolis-Hastings algorithm in  \cite{Hastings1970}, $p_{ii} = 1 - \sum_{j \neq i} p_{ij}$.

\begin{algorithm}
  \SetKwInOut{Input}{input}\SetKwInOut{Output}{output}
  \SetKw{Return}{return}
  \Input{Current and target configuration $\mathcal{C}$ and $\mathcal{C}_T$}
  \While{True} {
    Randomly pick an agent $k$ in state $x_i$
    
    Compute restricted action set $\mathcal{R}_k$
    
    Select $a_{i \rightarrow j, k} \in \mathcal{R}_k$ with probability $q_{ij} = \frac{1}{|\mathcal{R}_k|}$
    
    Compute $\alpha_{ij} = \min \left\{1, \frac{q_{ji}}{q_{ij}} e^{\frac{1}{\tau} (\Phi(x_j)- \Phi(x_i))} \right\}$
    
    \eIf{$\alpha_{ij} = 1$} {
      $x_{t+1} = x_j$
    }{
       $x_{t+1} =  \
	\begin{cases} 
	    x_{j} & \mbox{with probability } \alpha_{ij} \\
	    x_{i} & \mbox{with probability } 1 - \alpha_{ij}
	\end{cases}$
    }
  }
  \caption{Global game-theoretic learning algorithm. Note that state $x_j$ is the result of agent $k$ applying action $a_{i \rightarrow j, k}$ and $x_i$ and $x_j$ refer to states of the entire configuration. Also note that $q_{ii} \notin \mathcal{R}_k$ but $p_{ii} \neq 0$.}
  \label{alg:global}
\end{algorithm}

\section{A decentralized Algorithm}
\label{sec:decentralizedAlgorithm}
One shortcoming of Algorithm \ref{alg:global} is its centralized nature that requires the computation of a global potential function $\Phi(x_i)$ and depends on the entire current configuration $x_i$. A decentralized algorithm is desirable to execute the global learning rule on a team of agents. The formulation of the self-reconfiguration problem as a potential game allows us to rewrite the transition probabilities in a decentralized fashion as follows.

\begin{corollary}
\label{cor:decentralizedTransitionProbabilities}
 The global learning rule of Theorem \ref{theorem:transition_probabilities} can be decentralized such that each agent can execute it with local information only.
\end{corollary}

\begin{proof}
 Note that for agent $k$, we can express $x_j$ and $x_i$ as $(a_k', a_{-k})$ and $(a_k, a_{-k})$ respectively.  According to Proposition \ref{prop:potentialGame}, we can then rewrite $\Phi(x_j) - \Phi(x_i)$ as follows.
 \begin{eqnarray}
    U_k(a_k') - U_k(a_k) &=& U_k(a_k', a_{-k}) - U_k(a_k, a_{-k}) \nonumber \\
			 &=&\Phi(a_k', a_{-k}) - \Phi(a_k, a_{-k}) \nonumber \\
			 &=& \Phi(x_j) - \Phi(x_i) \nonumber
 \end{eqnarray}
 Therefore, we can rewrite the transition probabilities as follows.
 $$p_{ij} =  \
 \begin{cases} 
  q_{ji}e^{\frac{1}{\tau}(U_k(a_k') - U_k(a_k))} & \mathrm{if} \; e^{\frac{1}{\tau}(U_k(a_k') - U(a_k))}\frac{q_{ji}}{q_{ij}} \leq 1 \\
  q_{ij} & \mathrm{o.w.}
 \end{cases}$$
 Since $q_{ij}$, $q_{ji}$, $U_k(a_k')$, as well as $U_k(a_k)$ can be computed with local information, so too can the transition probabilities. The stationary distribution of the process described by $p_{ij}$ is the same Gibbs distribution as in Theorem \ref{theorem:transition_probabilities}.
\end{proof}

Note that local can mean multiple hops, because the computation of restricted action sets requires to maintain groundedness of all neighboring agents. Verifying groundedness requires a DFS search to the ground plane, which in the worst case can take $N - 1$ hops.

Algorithm \ref{alg:local} shows a decentralized implementation of Algorithm \ref{alg:global} and Corollary \ref{cor:decentralizedTransitionProbabilities}. Note that a transition from configuration $x_i$ to $x_j$ is accomplished by agent $k$ executing action $a' \in \mathcal{R}_k$ starting at its current location $a$. Therefore, we can interpret $q_{ij}$ as a transition probability for a forward action and $q_{ji}$ as a reverse action (see Fig. \ref{fig:forward_and_reverse_action_sets}). 

\begin{algorithm}
  \SetKwInOut{Input}{input}\SetKwInOut{Output}{output}
  \SetKw{Return}{return}
  \Input{Target configuration $\mathcal{C}_T$}
  Start clock (see \cite{Boyd2006})
  
  \While{True} {
    \If{Clock ticks} {
      Compute current restricted action set $\mathcal{R}_k$
      
      Select $a' \in \mathcal{R}_k$ with probability $q = \frac{1}{|\mathcal{R}_k|}$
      
      Compute $\alpha = \min \left\{1, \frac{|\mathcal{R}_k|}{|\mathcal{R}_{k'}|} e^{\frac{1}{\tau} (U(a')- U(a))} \right\}$
      
      \eIf{$\alpha = 1$} {
	$x_{t+1} = a'$
      }{
	$x_{t+1} =  \
	  \begin{cases} 
	      a' & \mbox{with probability } \alpha \\
	      a  & \mbox{with probability } 1 - \alpha
	  \end{cases}$
      }
    }
  }
  \caption{Self-reconfiguration using game-theoretic learning local version that each agent executes. Note that $x_t$, $x_{t+1}$ refer to consecutive states of the active agent.}
  \label{alg:local}
\end{algorithm}
%
\section{Implementation}
\label{sec:implementation}
The global version of the algorithm was implemented and evaluated in Matlab. For the simulations shown in Fig. \ref{fig:convergenceTimes}, we used a $\tau = 0.001$ that struck a balance between greedy maximization of agent utilities and exploration of the state space through suboptimal actions. 
Agents' actions are restricted according to the action set computation outlined in Section \ref{sec:actionSetComputation} and the motion model in Section \ref{sec:motionModel}. Restricted action sets depend on the agents joint action and the environment - in our simulations only the ground plane. Agents' positions are initialized above the ground plane such that their z-coordinate $z \ge 1$. In a straightforward extension to this algorithm, obstacles can be added to restrict the environment even further. 

Fig. \ref{fig:convergenceTimes} shows convergence results of Algorithm \ref{alg:global} of configurations containing 10, 20, and 30 agents. Four types of reconfigurations have been performed: 2D to 2D, 2D to 3D, 3D to 2D, and 3D to 3D. The vertical lines in Fig. \ref{fig:convergenceTimes} represent the average time to convergence of all four types of reconfigurations of a certain size (e.g. the leftmost line represents average convergence of a configuration of 10 agents). Convergence is achieved, if the configuration reaches a global potential of $\Phi = N$, i.e. every agent has a utility of $U_i = 1$. Note that in the scenarios of Fig. \ref{fig:convergenceTimes}, the target configuration was offset from the initial configuration by a translation of 10 units along the x-axis. One can observe that at the beginning of each reconfiguration the global potential ramps up very fast (within a few hundred time steps), but the asymptotic convergence to the global optimum can be slow (see the case 3D to 2D for 30 agents). An example of a 2D to 2D reconfiguration sequence is shown in Fig. \ref{fig:reconfigurationExample}.

\begin{figure*}
 \centering
 \includegraphics[width=1\textwidth]{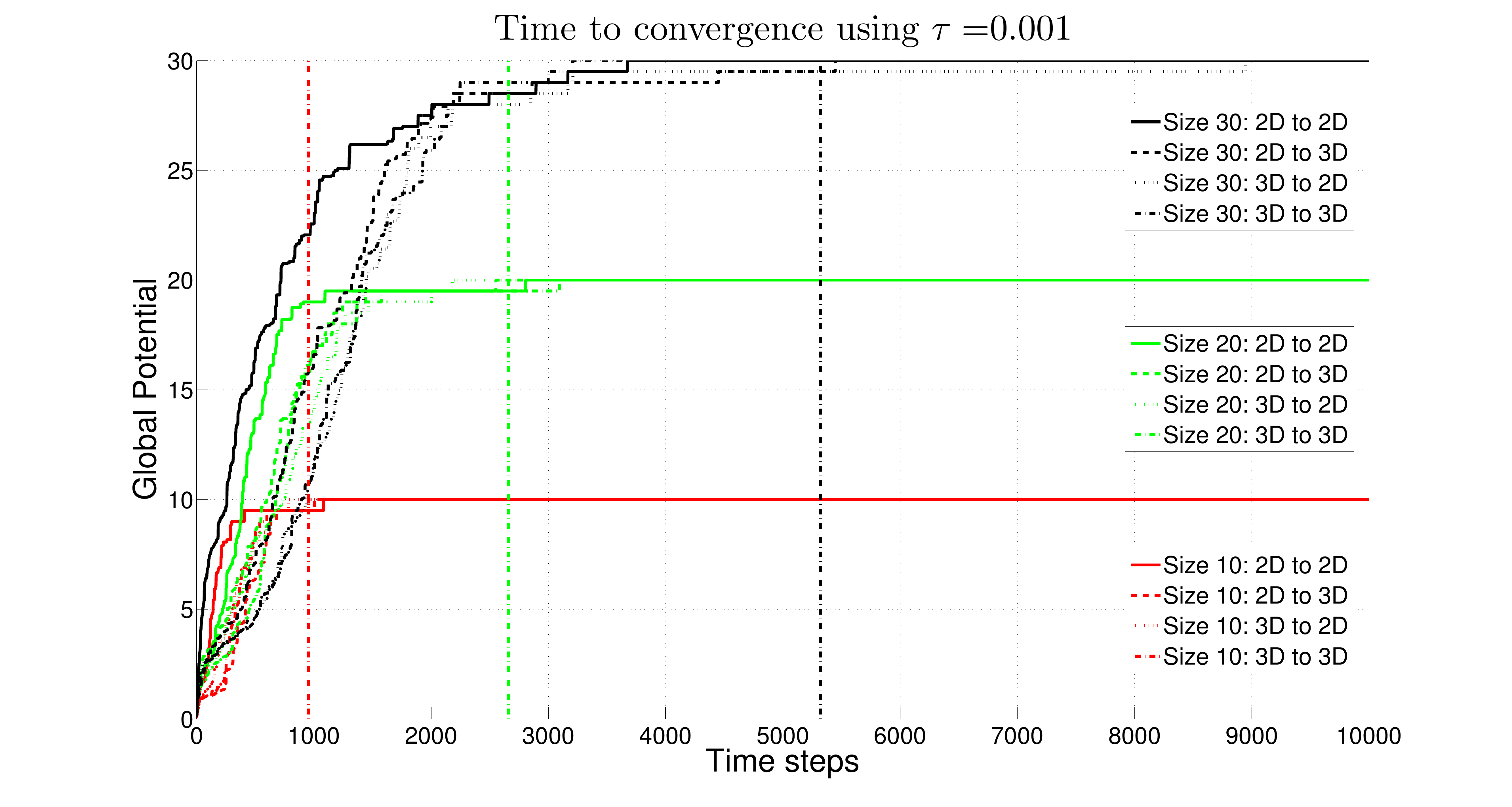}
 \caption{Convergence times for different types of configurations and sizes ranging from 10 to 30 agents.}
 \label{fig:convergenceTimes}
\end{figure*}

\section{Conclusion}
\label{sec:conclusion}
In this paper, we have applied the potential game formalism to the self-reconfiguration problem and developed a stochastic algorithm that converges to the global potential function maximizer. Under the assumptions of groundedness, we have introduced a decentralized approach to self-reconfiguration that requires little communication, does not rely on a centralized decision maker or precomputation, and converges to the global optimum. Our simulation results suggest that this formulation of the self-reconfiguration problem is indeed a feasible approach and can solve self-reconfiguration in a decentralized fashion. 

\bibliographystyle{plain}


\begin{thebibliography}{10}

\bibitem{Arslan2007}
G{\"u}rdal Arslan, Jason~R Marden, and Jeff~S Shamma.
\newblock Autonomous vehicle-target assignment: A game-theoretical formulation.
\newblock {\em Journal of Dynamic Systems, Measurement, and Control},
  129(5):584--596, September 2007.

\bibitem{Blume1993}
Lawrence~E. Blume.
\newblock The statistical mechanics of strategic interaction.
\newblock {\em Games and economic behavior}, 5(3):387--424, 1993.

\bibitem{Boyd2006}
Stephen~P. Boyd, Arpita Ghosh, Balaji Prabhakar, and Devavrat Shah.
\newblock Randomized gossip algorithms.
\newblock {\em IEEE Transactions on Information Theory}, 52(6):2508--2530,
  2006.

\bibitem{Butler2001}
Zack Butler, Keith Kotay, Daniela Rus, and Kohji Tomita.
\newblock Cellular automata for decentralized control of self-reconfigurable
  robots.
\newblock In {\em Proc. of the ICRA 2001 workshop on modular robots}, pages
  21--26, 2001.

\bibitem{Butler2004}
Zack Butler, Keith Kotay, Daniela Rus, and Kohji Tomita.
\newblock Generic decentralized control for lattice-based self-reconfigurable
  robots.
\newblock {\em The International Journal of Robotics Research}, 23(9):919--937,
  2004.

\bibitem{Cheung2011}
Kenneth~C. Cheung, Erik~D. Demaine, Jonathan Bachrach, and Saul Griffith.
\newblock Programmable assembly with universally foldable strings (moteins).
\newblock {\em IEEE Transactions on Robotics}, 27(4):718--729, 2011.

\bibitem{Fitch2008}
Robert Fitch and Zack Butler.
\newblock Million module march: Scalable locomotion for large
  self-reconfiguring robots.
\newblock {\em The International Journal of Robotics Research},
  27(3-4):331--343, 2008.

\bibitem{Fitch2003}
Robert Fitch, Zack Butler, and Daniela Rus.
\newblock Reconfiguration planning for heterogeneous self-reconfiguring robots.
\newblock In {\em Intelligent Robots and Systems, 2003.(IROS 2003).
  Proceedings. 2003 IEEE/RSJ International Conference on}, volume~3, pages 2460
  -- 2467, October 2003.

\bibitem{Fox2012}
Michael~J. Fox.
\newblock {\em Distributed Learning in Large Populations}.
\newblock PhD thesis, Georgia Institute of Technology, August 2012.

\bibitem{Hastings1970}
W~Keith Hastings.
\newblock Monte carlo sampling methods using markov chains and their
  applications.
\newblock {\em Biometrika}, 57(1):97--109, 1970.

\bibitem{Kotay2005}
Keith Kotay and Daniela Rus.
\newblock Efficient locomotion for a self-reconfiguring robot.
\newblock In {\em Robotics and Automation, 2005. ICRA 2005. Proceedings of the
  2005 IEEE International Conference on}, pages 2963--2969, 2005.

\bibitem{Kurokawa2008}
Haruhisa Kurokawa, Kohji Tomita, Akiya Kamimura, Shigeru Kokaji, Takashi Hasuo,
  and Satoshi Murata.
\newblock Distributed self-reconfiguration of {M-TRAN III} modular robotic
  system.
\newblock {\em The International Journal of Robotics Research},
  27(3-4):373--386, 2008.

\bibitem{Lim2011}
Yusun~Lee Lim.
\newblock Potential game based cooperative control in dynamic environments.
\newblock Master's thesis, Georgia Institute of Technology, 2011.

\bibitem{Marden2012}
Jason~R. Marden and Jeff~S. Shamma.
\newblock Revisiting log-linear learning: Asynchrony, completeness and
  payoff-based implementation.
\newblock {\em Games and Economic Behavior}, 75(2):788--808, 2012.

\bibitem{Metropolis1953}
Nicholas Metropolis, Arianna~W. Rosenbluth, Marshall~N. Rosenbluth, Augusta~H.
  Teller, and Edward Teller.
\newblock Equation of state calculations by fast computing machines.
\newblock {\em Journal of Chemical Physics}, 21:1087--1092, 1953.

\bibitem{Monderer1996}
Dov Monderer and Lloyd~S Shapley.
\newblock Potential games.
\newblock {\em Games and Economic Behavior}, 14(1):124--143, 1996.

\bibitem{Pickem2012}
Daniel Pickem and Magnus Egerstedt.
\newblock Self-reconfiguration using graph grammars for modular robotics.
\newblock In {\em Analysis and Design of Hybrid Systems (ADHS), 4th IFAC
  Conference on}, Eindhoven, Netherlands, June 2012.

\bibitem{Pickem2013}
Daniel Pickem, Magnus Egerstedt, and Jeff~S Shamma.
\newblock Complete heterogeneous self-reconfiguration: Deadlock avoidance using
  hole-free assemblies.
\newblock In {\em Distributed Estimation and Control in Networked Systems
  (NecSys'13), 4th IFAC Workshop on}, volume~4, pages 404--410, September 2013.

\bibitem{Rus2001}
Daniela Rus and Marsette Vona.
\newblock Crystalline robots: Self-reconfiguration with compressible unit
  modules.
\newblock {\em Autonomous Robots}, 10(1):107--124, Jan. 2001.

\bibitem{Young1993}
Peyton~H. Young.
\newblock The evolution of conventions.
\newblock {\em Econometrica}, January 1993.

\bibitem{Zhu2013}
Minghui Zhu and Sonia Mart\'{\i}nez.
\newblock Distributed coverage games for energy-aware mobile sensor networks.
\newblock {\em SIAM Journal on Control and Optimization}, 51(1):1--27, 2013.

\end{thebibliography}

\end{document}